\documentclass[a4paper,reqno]{amsart}
\usepackage[all]{xy}           
\usepackage{amssymb}           
\usepackage{hyperref}
\usepackage{eucal}
\usepackage{graphicx}
\usepackage{epsfig}
\usepackage[usenames,dvipsnames]{color}
\usepackage{subfigure}
\usepackage{color}
\numberwithin{equation}{section}

\newtheorem{definition}{Definition}[section]

\newtheorem{theorem}[definition]{Theorem}

\newtheorem{remarkth}[definition]{Remark}

\renewcommand{\emph}[1]{{\bfseries\itshape{#1}}}

\numberwithin{figure}{section}

\newcommand{\R}{\mathbb{R}}      

\newcount\ancho \newcount\anchom \newcount\anchoa
\newcount\anchob \newcount\altura

\newcommand{\ltilde}[3][0]{\altura=0 \advance\altura by #1
           \ancho=#2 \anchom=\ancho \divide\anchom by 2
           \anchoa=\ancho \divide\anchoa by 4
           \anchob=\anchom \advance\anchob by \anchoa
           \kern-3pt \begin{array}[b]{c}
           \begin{picture}(1,1)(\anchom,-\altura)
        \qbezier(0,2)(\anchoa,5)(\anchom,2)
        \qbezier(\anchom,2)(\anchob,-1)(\ancho,4)
        \qbezier(0,2)(\anchoa,4.5)(\anchom,1.8)
        \qbezier(\anchom,1.8)(\anchob,-1.5)(\ancho,4)
       \end{picture} \\[-4pt]{#3}
                       \end{array} \kern-4pt    }

\newcommand{\lhat}[3][0]{\altura=0 \advance\altura by #1
           \ancho=#2 \anchom=\ancho \divide\anchom by 2
           \anchoa=\ancho \divide\anchoa by 4
           \anchob=\anchom \advance\anchob by \anchoa
           \kern-3pt \begin{array}[b]{c}
           \begin{picture}(1,1)(\anchom,-\altura)
        \qbezier(0,2)(\anchoa,4)(\anchom,6)
        \qbezier(\anchom,6)(\anchob,4)(\ancho,2)
        \qbezier(0,2)(\anchoa,3.8)(\anchom,5.6)
        \qbezier(\anchom,5.6)(\anchob,3.8)(\ancho,2)
       \end{picture} \\[-4pt] {#3}
                       \end{array} \kern-4pt    }

\newcommand{\I}{I\mkern-7muI}

\makeatletter
\newcommand\prol{\@ifstar{\@proldf}{\@prolpf}}  
\def\@prolpf{\@ifnextchar[{\@prolpf@wrt}{\@prolpf@}}
\def\@prolpf@wrt[#1]#2{\@ifnextchar[{\@prolpf@wrt@at{#1}{#2}}{\@prolpf@wrt@{#1}{#2}}}
\def\@prolpf@wrt@at#1#2[#3]{\prolsymbol^{#1}_{#3}#2}
\def\@prolpf@wrt@#1#2{\prolsymbol^{#1}#2}
\def\@prolpf@#1{\@ifnextchar[{\@prolpf@at{#1}}{\@prolpf@@{#1}}}
\def\@prolpf@at#1[#2]{\prolsymbol_{#2}#1}
\def\@prolpf@@#1{\prolsymbol#1}
\def\@proldf{\@ifnextchar[{\@proldf@wrt}{\@proldf@}}
\def\@proldf@wrt[#1]#2{\@ifnextchar[{\@proldf@wrt@at{#1}{#2}}{\@proldf@wrt@{#1}{#2}}}
\def\@proldf@wrt@at#1#2[#3]{\prolsymbol^{*#1}_{#3}#2}
\def\@proldf@wrt@#1#2{\prolsymbol^{*#1}#2}
\def\@proldf@#1{\@ifnextchar[{\@proldf@at{#1}}{\@proldf@@{#1}}}
\def\@proldf@at#1[#2]{\prolsymbol^*_{#2}#1}
\def\@proldf@@#1{\prolsymbol^*#1}
\def\prolsymbol{\mathcal{T}}
\makeatother

\newcommand\rmd{\mathrm{d}}

\begin{document}
{\Large

\title{The inhomogeneous Suslov problem}


\author[L. C. \ Garc\'{\i}a-Naranjo]{Luis C. \ Garc\'{\i}a-Naranjo}
\address{L. C.\ Garc\'{\i}a-Naranjo:
Departamento de Matem\'aticas y Mec\'anica \\
IIMAS-UNAM \\
Apdo Postal 20-726,  Mexico City,  01000, Mexico}
\email{luis@mym.iimas.unam.mx}

\author[A. J. \ Maciejewski ]{Andrzej J. Maciejewski }
\address{A. J. \ Maciejewski:
J. Kepler Institute of Astronomy \\
University of Zielona G\'ora \\
Licealna 9, 65-417 Zielona G\'ora, Poland}
\email{andrzej.j.maciejewski@gmail.com}

\author[J.\ C.\ Marrero]{Juan C.\ Marrero}
\address{Juan C.\ Marrero:
ULL-CSIC Geometr\'{\i}a Diferencial y Mec\'anica Geom\'etrica\\
Departamento de Matem\'atica Fundamental, Facultad de
Ma\-te\-m\'a\-ti\-cas, Universidad de la Laguna, La Laguna,
Tenerife, Canary Islands, Spain} \email{jcmarrer@ull.es}

\author[M. Przybylska]{Maria Przybylska}
\address{M. Przybylska:
Institute of Physics \\
University of Zielona G\'ora \\
Licealna 9, 65-417 Zielona G\'ora, Poland}
\email{maria.przybylska@gmail.com}

\thanks{This work has been partially supported by MEC (Spain)
 MTM2011-15725-E, MTM2012-34478 and the project of the Canary Government ProdID20100210. 
LGN acknowledges the hospitality at the  Departamento de Matem\'atica Fundamental,
at Universidad de la Laguna, for its hospitality in June 2013.
}



\begin{abstract}
We consider the Suslov problem of nonholonomic rigid body motion with inhomogeneous constraints. We show that if the direction along which the Suslov
constraint is enforced is perpendicular to a principal axis of inertia of the body, then the reduced equations are integrable and, in the generic case, possess a smooth
invariant measure. Interestingly, in this generic case, the  first integral that permits integration is transcendental and the density of the invariant measure depends on the angular
velocities. We also study the Painlev\'e property of the solutions.
\end{abstract}

\maketitle

\tableofcontents

\section{Definition of the problem}

Consider the motion of a rigid body under its own inertia subjected to
the constraint
\begin{equation*} {\bf a} \cdot {\boldsymbol \Omega}=K,
\end{equation*}
where $K\in \R$ is constant. In the above, ${\bf a}\in \R^3$ is a fixed unit vector
in the body frame and ${\boldsymbol \Omega}\in \R^3$ is the angular velocity
of the body also written in the body frame. In the case where $K=0$ we
recover the classical nonholonomic Suslov problem.

Apparently Suslov \cite{Suslov} suggested a mechanism to physically
implement such a constraint that is described in
\cite{BorMamKil}.

Denote by $\I$ the inertia tensor of the body. It is a symmetric,
positive definite $3\times 3$ matrix. The equations of motion are
obtained via the Lagrange d'Alembert principle that yields
\begin{equation}
  \label{E:motion}
  \I  \dot {\boldsymbol \Omega} =\I {\boldsymbol \Omega}\times {\boldsymbol \Omega} +\lambda {\bf a},
\end{equation}
where the Lagrange multiplier $\lambda$ is determined by the condition
that the constraint is satisfied and ``$\times$" denotes the vector product in $\R^3$.

Differentiating the constraint and using the equation of motion we
obtain
\begin{equation*}
  \lambda=-\frac{\left (\I {\boldsymbol \Omega}\times {\boldsymbol \Omega}  \right ) \cdot \I^{-1}{\bf a}}{ {\bf a} \cdot \I^{-1}{\bf a}}.
\end{equation*}

With the above choice of $\lambda$ the equations of motion
\eqref{E:motion} preserve the quantity ${\bf a} \cdot {\boldsymbol
  \Omega}$.  The physical system of interest is obtained by
considering the motion on the level set ${\bf a} \cdot {\boldsymbol
  \Omega}=K$.

Note that the energy of the system, $H=\frac{1}{2}\I{\boldsymbol
  \Omega} \cdot {\boldsymbol \Omega} $ is only preserved on the level
set $K=0$. The inhomogeneous constraint adds or takes away energy from
the system.

We will assume that the body frame is oriented in such way that the
vector ${\bf a}=(0,0,1)$. The constraint is then $\Omega_3=K$. Without
loss of generality, we can also assume that the entry $I_{12}$ of the
inertia tensor vanishes. Thus, the inertia tensor has the form
\begin{equation*}
  \I=\left ( \begin{array}{ccc} I_{11} & 0 & I_{13} \\ 0 & I_{22} & I_{23} \\I_{13}  & I_{23} & I_{33} \end{array}\right ).
\end{equation*}
In this case we find that the equations for $\Omega_1, \, \Omega_2$ on
the level set $\Omega_3=K$ are given by:
\begin{equation}
  \label{E:MotionGeneral}
  \begin{split}
    I_{11} \dot \Omega_1 &= -\Omega_2(I_{13}\Omega_1+I_{23}\Omega_2) + \Omega_2(I_{22}K-I_{33}K)+I_{23}K^2, \\
    I_{22}\dot \Omega_2 &=\Omega_1(I_{13}\Omega_1+I_{23}\Omega_2)
    +\Omega_1(-I_{11}K+I_{33}K)-I_{13}K^2.
  \end{split}
\end{equation}

The case where $K=0$ corresponds to the classical Suslov problem that
has been studied in detail. In this case there are two distinct cases
of qualitative motion.

\begin{enumerate}
\item If the vector $\bf a$ is an eigenvector of the inertia tensor
  $\I$, then $\I$ is diagonal ($I_{13}=I_{23}=0$) and the dynamics is
  trivial. The angular velocity is constant so the body rotates about
  a fixed axis with constant speed.
\item If the vector $\bf a$ is not an eigenvector of the inertia
  tensor $\I$, then the system possesses a straight line of asymptotic
  equilibria. Using the conservation of energy, the equations of
  motion are integrated in terms of hyperbolic functions. In this case
  there is no smooth invariant measure.  For a discussion of the
  motion of the body in this case see \cite{FMP}.

\end{enumerate}

In this note we consider the case where $K$ is non-zero. Note that
$\frac{1}{K}$ is a natural time scale for the system, so we introduce
the non-dimensional variables
\begin{equation*}
  \tau=Kt, \qquad \omega_1=\frac{1}{K}\Omega_1, \qquad \omega_2=\frac{1}{K}\Omega_2.
\end{equation*}
The system \eqref{E:MotionGeneral} becomes
\begin{equation}
  \label{E:MotionNondim}
  \begin{split}
    I_{11}  \omega_1' &= -\omega_2(I_{13}\omega_1+I_{23}\omega_2) + \omega_2(I_{22}-I_{33})+I_{23}, \\
    I_{22} \omega_2' &=\omega_1(I_{13}\omega_1+I_{23}\omega_2)
    +\omega_1(-I_{11}+I_{33})-I_{13},
  \end{split}
\end{equation}
where $'=\frac{d}{d\tau}$.

For the rest of the paper we will analyze the system
\eqref{E:MotionNondim} depending on the position of the vector $\bf a$
relative to the principal axes of inertia of the body. We consider two
cases, the simplest one when the vector $\bf a$ is an eigenvector of
the inertia tensor $\I$, and the second one, when ${\bf a}$ belongs to
a two-dimensional eigenspace of $\I$ but is not an eigenvector. The
analysis for a generic $\bf a$ will be postponed for a subsequent
publication.

\section{Suppose that ${\bf a}$ is an eigenvector of $\I$}
\label{S:a-eigenvector}

The simplest case of motion also occurs when ${\bf a}$ is an
eigenvector of $\I$. In this case $I_{13}=I_{23}=0$ and the equations
of motion become linear:
\begin{equation*}
  \begin{split}
    I_{11} \omega_1' &=  (I_{22}-I_{33})\omega_2, \\
    I_{22} \omega_2' &=(-I_{11}+I_{33})\omega_1.
  \end{split}
\end{equation*}
The trace of the associated constant matrix is zero and its
determinant equals
\begin{equation*}
  \frac{(I_{22}-I_{33})(I_{11}-I_{33})}{I_{11}I_{22}}.
\end{equation*}
The above determinant is greater than zero if either $I_{11},\, I_{22}
>I_{33}$ or $I_{11},\, I_{22} <I_{33}$. So we conclude that if ${\bf
  a}$ is an eigenvector of the inertia tensor, along the axis
corresponding to the largest or smallest moment of inertia, then we
have simple-harmonic motion in the $\omega_1, \omega_2$ plane.

Similarly, if ${\bf a}$ points along the axis of middle inertia, then
we have a linear saddle in the $\omega_1, \omega_2$ plane.  The
dynamics in the case where the body has rotational symmetry and some
of the principal moments of inertia coincide can be easily understood.

\section{Suppose that ${\bf a}$ belongs to a two-dimensional eigenspace of
  $\I$}

In this section we consider the case when the vector ${\bf
  a}$ belongs to the two-dimensional space spanned by two of the
principal axes of inertia of the body, but is not aligned with any of
them. This is equivalent to saying that the vector $\bf a$ is
perpendicular to a principal axis of inertia but without defining one
of them.

We suppose that $I_{13}=0$ but $I_{23}\neq 0$. Under these
assumptions, the principal moments of inertia of the body are
\begin{equation}
\label{E:Eigenvalues}
\begin{split}
  J_1=I_{11}, \qquad J_2=\frac{1}{2}(I_{22}+I_{33})+\frac{1}{2}\sqrt{(I_{22}-I_{33})^2+4I_{23}^2} , \\
  J_3=\frac{1}{2}(I_{22}+I_{33})-\frac{1}{2}\sqrt{(I_{22}-I_{33})^2+4I_{23}^2},
\end{split}
\end{equation}
and the vector ${\bf a}$ belongs to the two dimensional eigenspace of
$\I$ spanned by the principal axes of inertia of the body associated
to $J_2$ and $J_3$. In other words, ${\bf a}$ is orthogonal to the
principal axis of inertia associated to $J_1$.

The equations of motion \eqref{E:MotionNondim} simplify to:
\begin{equation}
  \label{E:Motion-simpler}
  \begin{split}
    J_1  \omega_1' &= -I_{23}\omega^2_2 + \omega_2(I_{22}-I_{33})+I_{23}, \\
    I_{22}\omega_2' &=\omega_1\left (I_{23}\omega_2+(I_{33}-J_1)
    \right ).
  \end{split}
\end{equation}

The system possesses a particular solution of the form
\begin{equation*}
  \label{E:partsol}
  \omega_2=\frac{J_1-I_{33}}{I_{23}}, \qquad \omega_1=-\frac{(J_1-J_2)(J_1-J_3)}{I_{23}J_1}t+c_0,
\end{equation*}
where $c_0$ is an arbitrary constant. Hence, the horizontal line
$\omega_2=\frac{J_1-I_{33}}{I_{23}}$ is invariant by the flow and so
are the semi-planes
\begin{equation*}
  \omega_2>\frac{J_1-I_{33}}{I_{23}} \qquad \mbox{and} \qquad \omega_2<\frac{J_1-I_{33}}{I_{23}}.
\end{equation*}
 
At this point, we divide our analysis in two separate cases depending
on whether $J_1$ coincides with either of $J_2$ or $J_3$, or not.

\subsection{Suppose that $J_1\neq J_2, J_3$.}
\label{SS:Unequal-Moments}
The system
\eqref{E:Motion-simpler} possesses the integral of motion
\begin{equation}
  \label{E:ConsQty}
  F(\omega_1,\omega_2)=(I_{23}\omega_2+I_{33}-J_1)\exp \left (\frac{I_{23}^2\left (J_1\omega_1^2+I_{22}\left (\omega_2-\frac{(I_{22}-J_1)}{I_{23}} \right )^2 \right )}{2I_{22}(J_1-J_2)(J_1-J_3)} \right ).
\end{equation}
The exponential dependence on the integral is remarkable considering
that the system is only polynomial. To our knowledge this is the first
example of a transcendental first integral in a polynomial mechanical
system.

The existence of the above integral implies that the system is
integrable by quadratures. It will be shown in Section \ref{S:Painleve} ahead that in this case there are solutions to
the system that  are multi-valued functions of complex time.

The system also possesses the following invariant measure:
\begin{equation*}
  \mu=\exp \left (\frac{I_{23}^2\left (J_1\omega_1^2+I_{22}\left (\omega_2-\frac{(I_{22}-J_1)}{I_{23}} \right )^2 \right )}{2I_{22}(J_1-J_2)(J_1-J_3)} \right )\, \mathrm{d}\omega_1 \, \mathrm{d}\omega_2.
\end{equation*}
Note that the density of the invariant measure is non trivial and
depends on the velocities. Something that is rare in mechanical
systems.

If $J_1$ is the middle moment of inertia, then the density decays to
zero as $(\omega_1, \omega_2)$ goes to infinity, and the measure of
the entire phase space is finite. On the other hand, if $J_1$ is the
largest or the smallest moment of inertia, the density goes to
infinity as $(\omega_1, \omega_2)$ goes to infinity, and the measure
of the entire phase space is infinite.

The system possesses exactly two equilibrium points located at
\begin{equation*}
  \omega_1=0, \qquad \omega_2^\pm=\left ( \frac{I_{22}-I_{33}\pm \sqrt{(I_{22}-I_{33})^2+4I_{23}^2}}{2I_{23}}\right ).
\end{equation*}
Under our assumption that $J_1\neq J_2, J_3$, none of these equilibria
lie on the line $\omega_2=\frac{J_1-I_{33}}{I_{23}}$.
 The eigenvalues of the Jacobi matrices at these equilibria are equal to 
$\pm\lambda_+$, and $\pm\lambda_-$,  where
\begin{equation}
\label{E:eigenvalues}
  \lambda_+^2=\frac{(J_1-J_2)(J_2-J_3)}{J_1I_{22}},\quad
  \lambda_-^2=\frac{(J_1-J_3)(J_3-J_2)}{J_1I_{22}}.
\end{equation}
 
 To understand the stability of the equilibria we fix the values of
$I_{22}, I_{33}$ and $I_{23}$ and use $J_1$ as a bifurcation
parameter.  
The bifurcation points correspond to $J_1=J_2$ and
$J_1=J_3$ where either $\lambda_+$ or $\lambda_-$ vanishes . The stability of the equilibria is easily determined
using the form of the  eigenvalues \eqref{E:eigenvalues} and the conserved
quantity \eqref{E:ConsQty}.  The global behaviour of trajectories
in the phase space $\omega_1 \omega_2$ is shown in Figure~\ref{fig:suslov13gen}.  The corresponding bifurcation
diagram is given in Figure~\ref{F:bifdiag} under the assumption that $J_2-J_3<J_3<I_{33}<I_{22}<J_2<J_2+J_3$. Note that,
by the triangle inequality for the principal moments of inertia, we have
the physical restriction for the values of $J_1$:
\begin{equation*}
  J_2-J_3\leq J_1 \leq J_2+J_3.
\end{equation*}

\begin{figure}[htp]
  \centering \subfigure[$J_1=1.6$]{
    \includegraphics[scale=0.44]{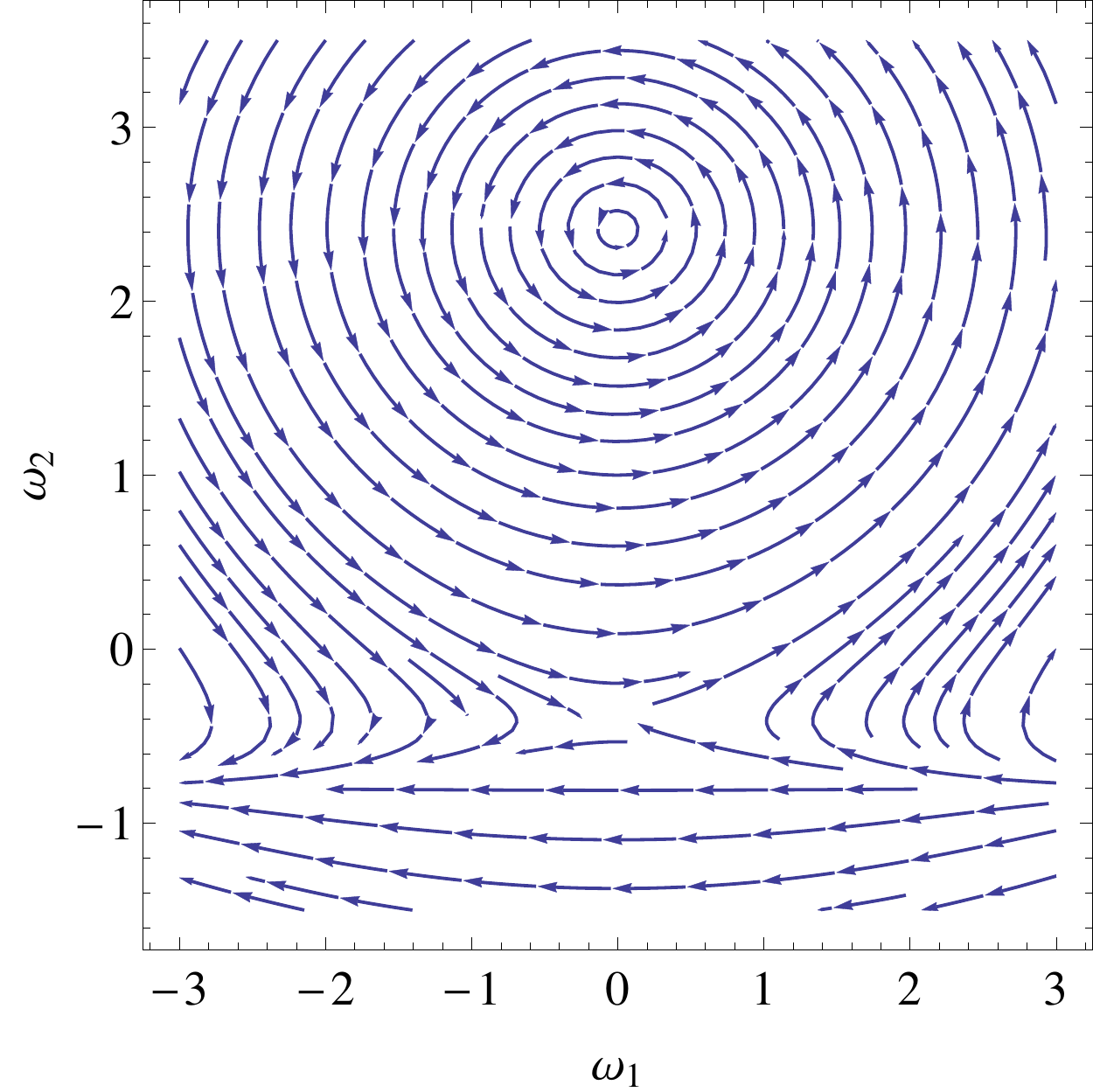}
    \label{fig:suslov13gen1}
  } \subfigure [$J_1=2.5$]{
    \includegraphics[scale=0.44]{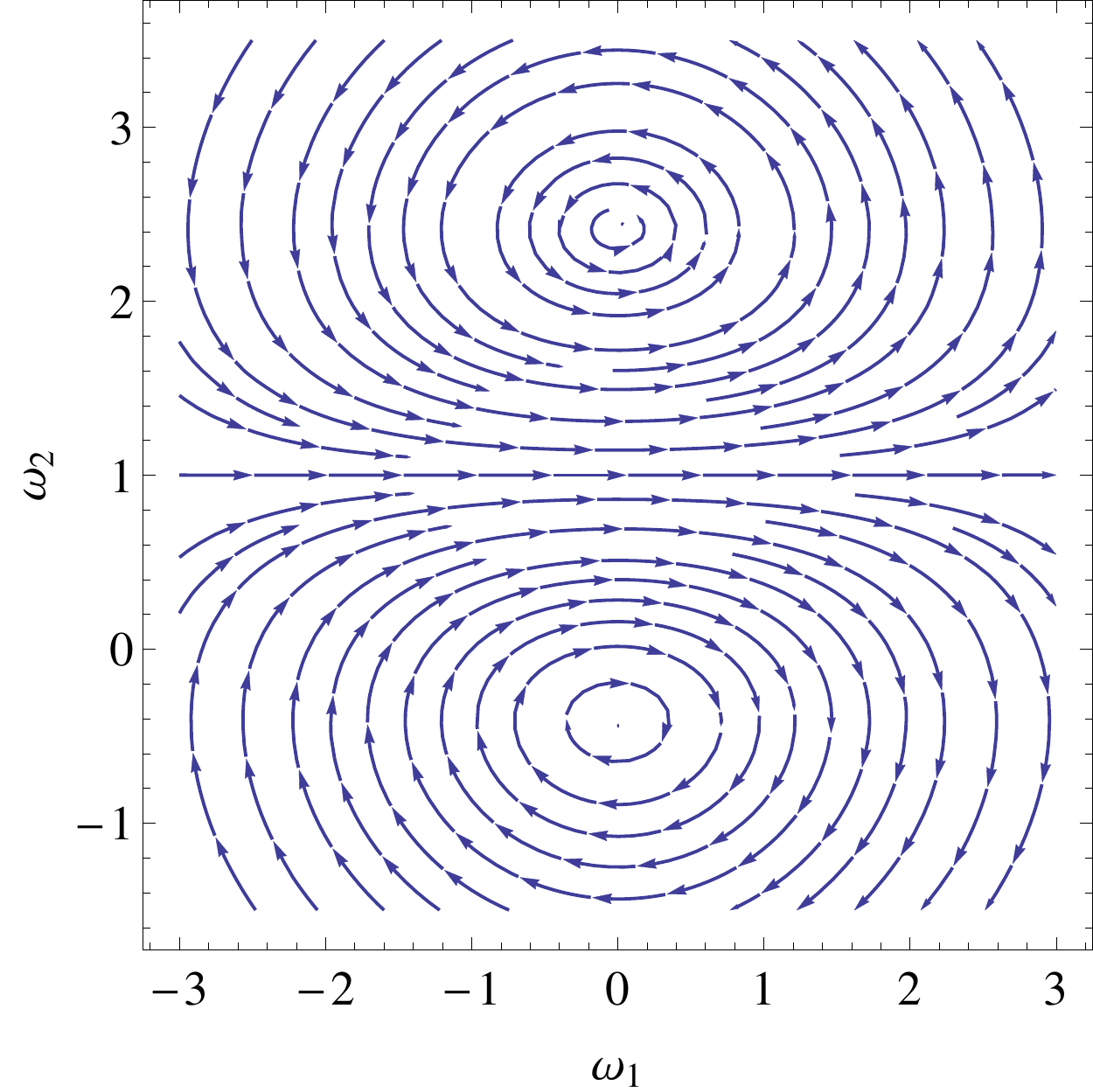}
    \label{fig:suslov13gen2}
  } \subfigure [$J_1=4$]{
    \includegraphics[scale=0.44]{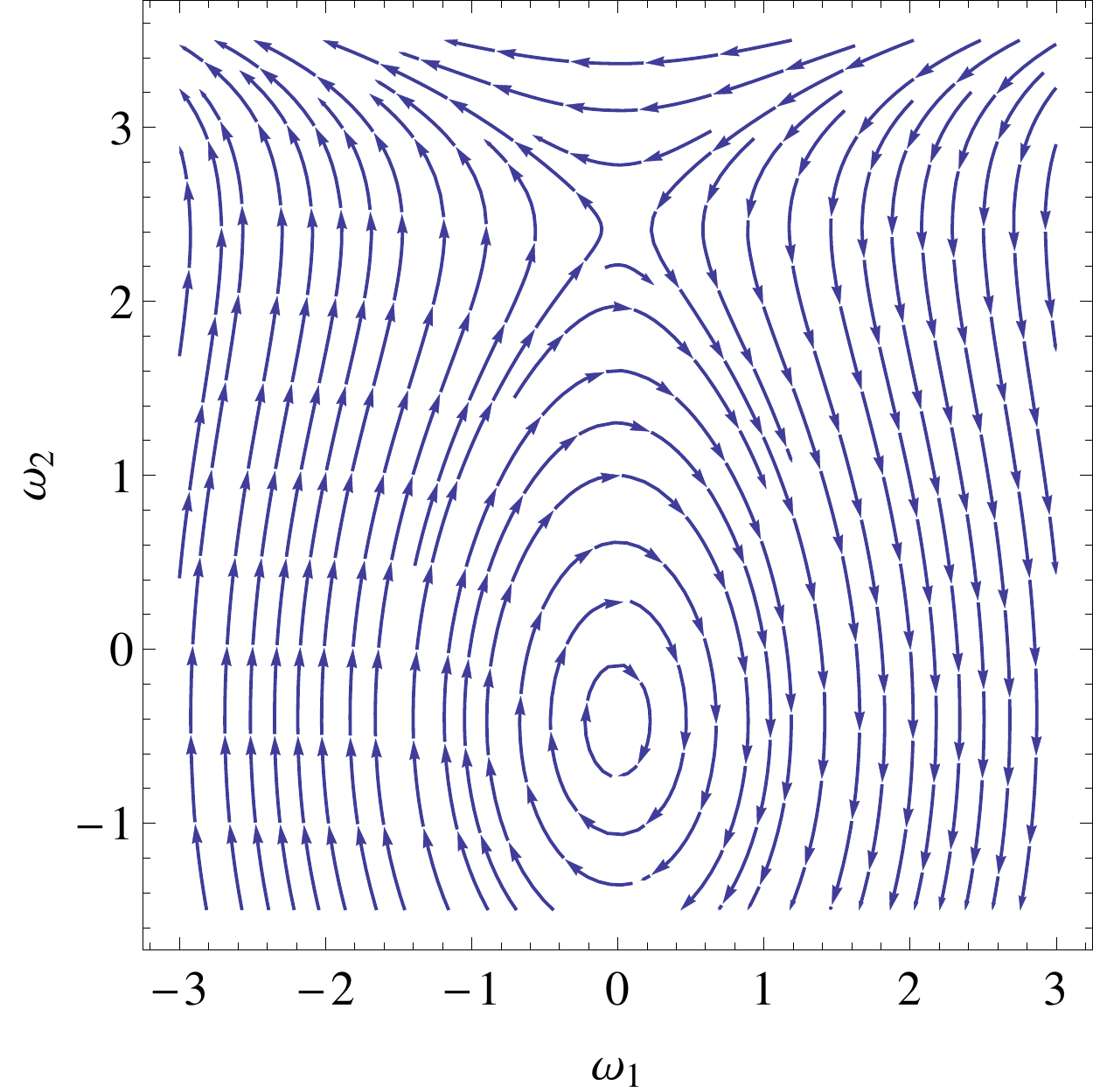}
    \label{fig:suslov13gen2}
  }
  \caption{Phase portraits of  affine Suslov problem with $J_2\neq
    J_1$ and $J_3\neq J_1$. Values of parameters: $I_{22}=3$,
    $I_{33}=2$, $I_{23}=1/2$. \label{fig:suslov13gen}}
\end{figure}

\begin{figure}[ht]
  \centering
  \includegraphics[width=10cm]{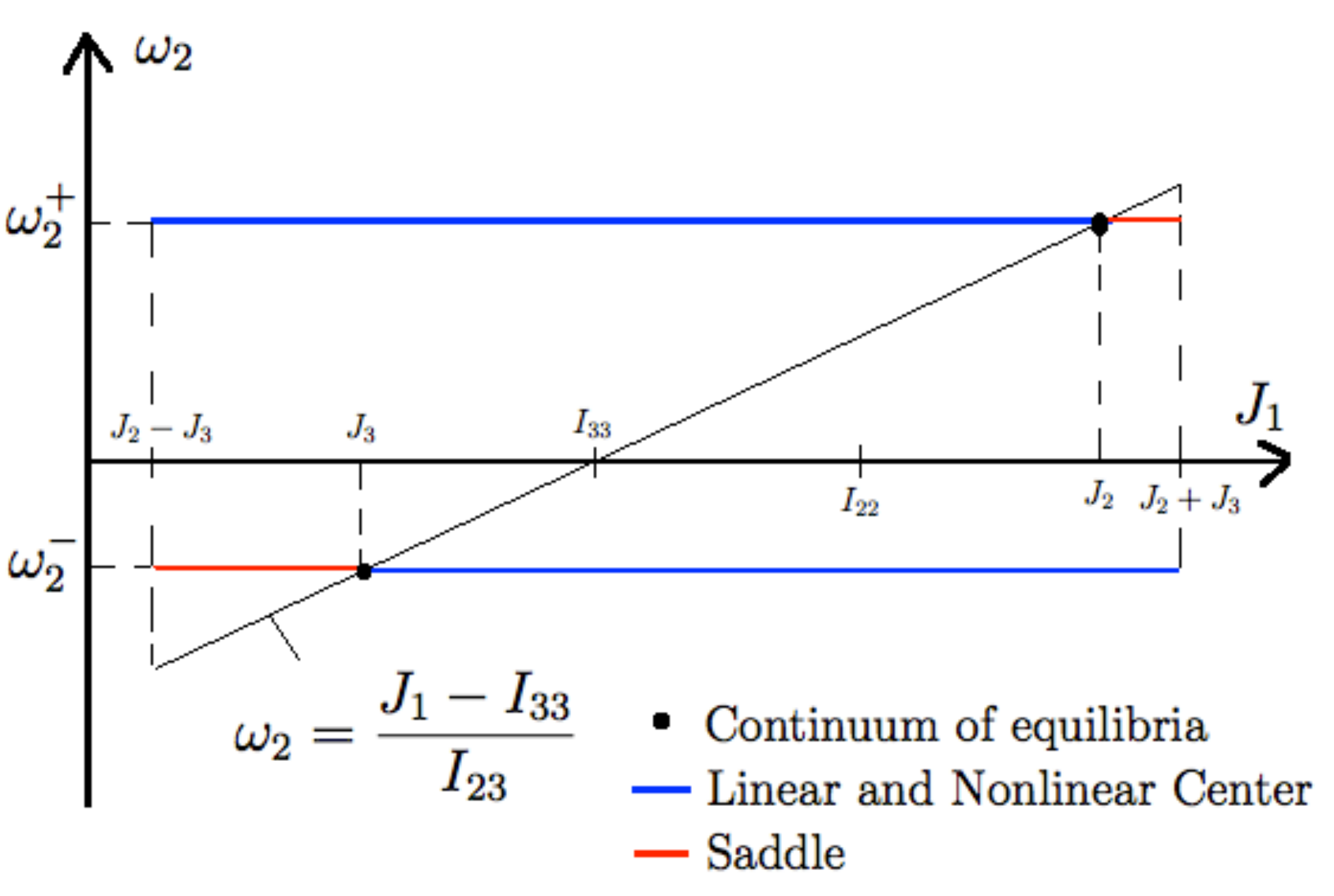}
  \caption{{Bifurcation diagram under the assumption
      $J_2-J_3<J_3<I_{33}<I_{22}<J_2<J_2+J_3$. This is physically
      attainable (for example if $I_{22}=3, \, I_{33}=2, \,
      I_{23}=1/2$).}\label{F:bifdiag} }
\end{figure}

\subsection{Suppose that $J_1= J_2$ or $J_1=J_3$.} 
\label{SS:EqualMoments}
 Define the quantity
\begin{equation}
\label{E:DefP}
P=I_{23}^2 - (J_1 - I_{22}) (J_1 - I_{33}).
\end{equation}
In view of \eqref{E:Eigenvalues}, the condition that $J_1= J_2$ or $J_1=J_3$ is equivalent to saying that  $P=0$.
In this section we will also work under  the assumption that $J_1\neq I_{33}$ because otherwise, the condition that
$P=0$ implies $I_{23}=0$ that brings us back to the case discussed in Section \ref{S:a-eigenvector}.
Therefore we can write
\begin{equation}
  I_{22}=\frac{J_{1}^2-I_{23}^2-J_{1}I_{33}}{J_{1}-I_{33}}.
  \label{eq:I22}
\end{equation}
Under this assumption, $J_1$ is obviously a multiplicity two eigenvalue of the inertia tensor. The other
eigenvalue of $\I$ is given in terms of $J_1, \, I_{33}$ and $I_{23}$ by
\begin{equation}
\label{E:ThirdEigenvalue}
J_0=\frac{J_1^2-J_1I_{33}-I_{23}^2}{J_1-I_{33}}.
\end{equation}

Substitution of \eqref{eq:I22} into \eqref{E:Motion-simpler} yields the following set of equations that possess a common
factor
\begin{equation}
  \omega'_1=\dfrac{(I_{23}\omega_2+I_{33}-J_1)(I_{33}\omega_2-J_1\omega_2-I_{23})}{J_1(J_1-I_{33})},\qquad \omega_2'=\frac{(J_1-I_{33})\omega_1(I_{23}\omega_2+I_{33}-J_1)}{J_1^2-I_{23}^2-J_1I_{33}}.
  \label{eq:spec}
\end{equation}

Notice that under our assumption that $J_1=J_2$ or $J_1=J_3$, the first
integral~\eqref{E:ConsQty} becomes indeterminate. However, in this case, there  exists
another one, quadratic in $\omega_1$ and $\omega_2$. We can obtain it
 using separation of variables in \eqref{eq:spec}.
 Namely, we can write
\[
\dfrac{\mathrm{d}\omega_1}{\mathrm{d}\omega_2}=\dfrac{(I_{33}\omega_2-J_1\omega_2-I_{23})(J_1^2-I_{23}^2-J_1I_{33})}{J_1(J_1-I_{33})^2\omega_1}.
\]
Now we use separation of variables to get
\[
 J_1(J_1-I_{33})^2\omega_1\, \mathrm{d}\omega_1 =
(I_{33}\omega_2-J_1\omega_2-I_{23})(J_1^2-I_{23}^2-J_1I_{33})\, \mathrm{d}\omega_2
\]
and we integrate both sides independently. After  multiplication by
two we obtain first integral
\begin{equation*}
\tilde  G(\omega_1,\omega_2)=J_1 (J_1 - I_{33})^2 \omega_1^2+(J_1 - I_{33}) (J_1^2 - I_{23}^2 - J_1 I_{33})\omega_2^2+2 I_{23} (J_1^2 - I_{23}^2 - J_1 I_{33})\omega_2.
  \end{equation*}
Notice that the coefficient of $\omega_2^2$ can be written as $I_{22}(J_1 - I_{33})^2$ and is therefore
positive (since the matrix $\I$ is positive definite). It follows that $\tilde G$ is positive definite and its
level sets are ellipses in the $\omega_1 \, \omega_2$ plane.
In order to integrate the system explicitly we perform a change of variables that puts the equations  in a simpler form. We introduce the variables $\xi, \eta$ by the relations
\begin{equation}
\label{E:Change-variables}
\omega_1=\frac{\xi}{I_{23}}, \qquad \omega_2=\frac{\eta}{I_{23}}+\frac{J_1-I_{33}}{I_{23}}.
\end{equation}
Then, the system \eqref{eq:spec} takes the form
\begin{equation}
\label{E:SimpEq}
\xi'=-A(\eta-B)\eta, \qquad \eta'=C\xi \eta,
\end{equation}
where the  constants $A, \, B, \, C$ are given by
\begin{equation*}
A=\frac{1}{J_1}, \qquad B=J_0-J_1, \qquad C=\frac{1}{I_{22}},
\end{equation*}
with $I_{22}$ and $J_0$ given respectively by \eqref{eq:I22} and \eqref{E:ThirdEigenvalue}.
The conserved quantity takes the simple form
\begin{equation*}
G(\xi,\eta)=A(\eta-B)^2+C\xi^2.
\end{equation*}
Under our assumptions, we have $A, \, C>0$ and $B\neq 0$. Notice that, similar to the classical 
Suslov problem, the system \eqref{E:SimpEq} 
possesses a continuum  of equilibrium points along the line $\eta=0$. On the other 
hand, the level sets of the conserved quantity $G$ are ellipses centered at the point $(0,B)$ in the
$\xi \eta$ plane, that is itself another equilibrium of the system. The number of intersections of these ellipses with the line $\eta=0$ will depend 
on the specific value of $G$. Let 
\begin{equation*}
G_0=AB^2=\frac{(J_0-J_1)^2}{J_1}>0.
\end{equation*}
Then, 
\begin{enumerate}
\item For $0<G<G_0$ there are no intersections of the level sets of $G$ with the line $\eta=0$.
The solutions are periodic and we shall see that they are expressed as a ratio of trigonometric functions.
\item The level set $G=G_0$  touches the line $\eta=0$ with multiplicity two at  the equilibrium point $(0,0)$. The level set is an orbit  homoclinic  to $(0,0)$ and we will see that 
the solution along this orbit is a rational function of $\tau$.
\item For values $G>G_0$ the ellipses cut the line $\eta=0$ at the  two equilibrium points \linebreak 
$\left (\pm \sqrt{\frac{G-AB^2}{C}}, 0 \right )$. The arcs of the ellipse connecting these points are two 
 heteroclinic orbits. The solutions along these orbits are expressed in terms of hyperbolic functions.
\end{enumerate}
The level set $G=0$ obviously consists of the individual equilibrium point $(0,B)$.
A schematic picture of the phase portrait is shown in Figure  \ref{F:phaseport}  below under the
assumption that $B>0$.
\begin{figure}[ht]
  \centering
  \includegraphics[width=8cm]{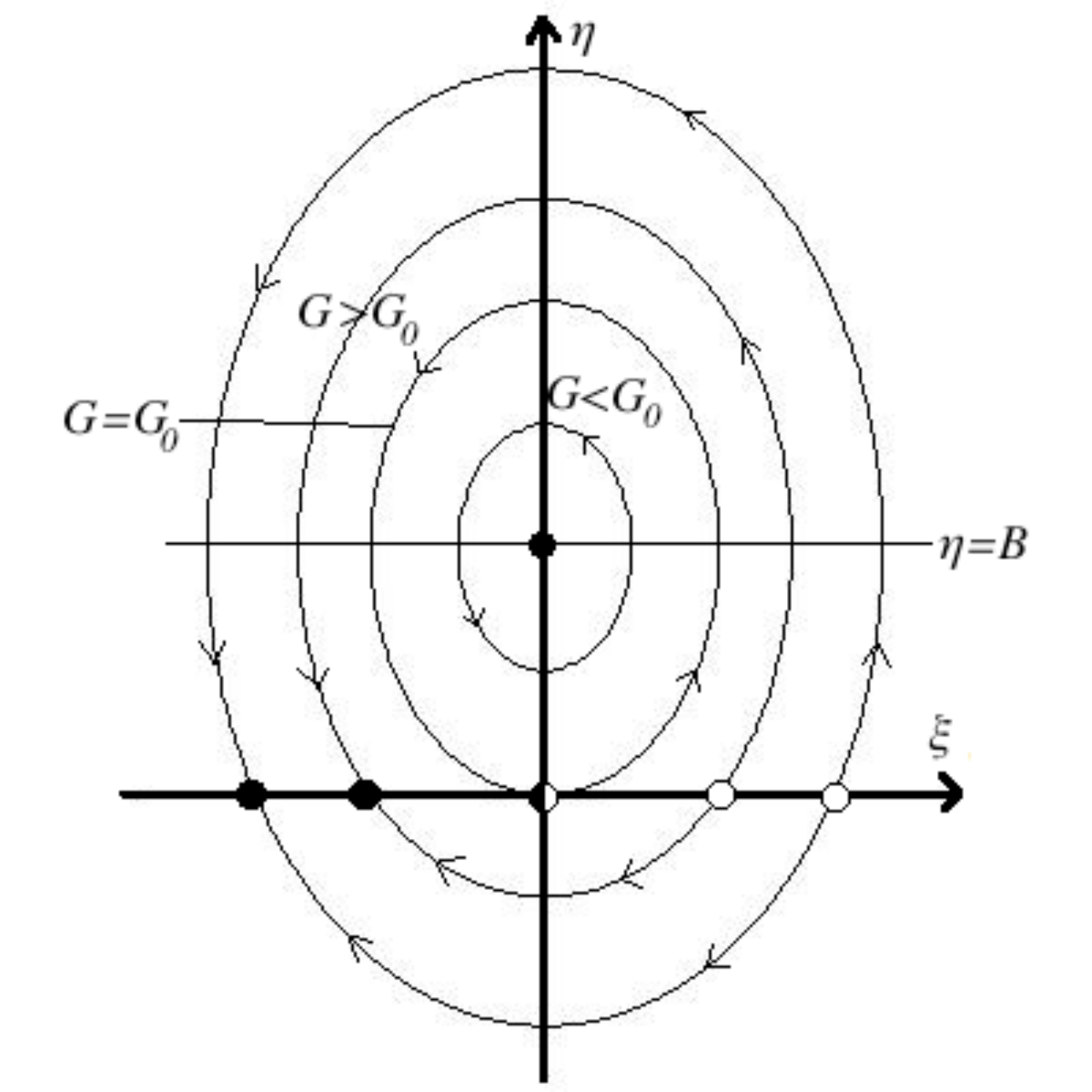}
  \caption{{Schematic representation of the phase portrait of equations  \eqref{E:SimpEq} under the assumption that
  $B>0$.}\label{F:phaseport} }
\end{figure}

{\bf Explicit solutions to \eqref{E:SimpEq}.}

The dependence of $\xi$ and $\eta$ on $\tau$ along the contour line $G(\xi,\tau)=g$ can be
obtained by introducing the parametrization:
\begin{equation}
\label{E:param}
\xi = \sqrt{\frac{g}{C}}\frac{2\psi}{\psi^2+1}, \qquad \eta =\pm \sqrt{\frac{g}{A}}\left ( \frac{1-\psi^2}{\psi^2+1}
\right ) + B.
\end{equation}
Substitution into \eqref{E:SimpEq} yields the separable equation for $\psi(\tau)$
\begin{equation*}
\psi'=K_1\psi^2+K_2
\end{equation*}
for certain constants $K_1$ and $K_2$ that satisfy $\mbox{sign}(K_1K_2)=\mbox{sign}(G_0-g)$.
Hence, as expected, the form of the solutions will depend on how $g$ compares to $G_0=AB^2$.

For $0<g<G_0$ it is enough to consider one branch of the parametrization \eqref{E:param}. Using
the ``-" branch and simplifying the algebra, one ends up with the explicit solution
\begin{equation*}
\begin{split}
\xi(\tau)&=\sqrt{\frac{g}{C}}\left ( \frac{\sqrt{AB^2-g}\sin ( \sqrt{C}\sqrt{AB^2-g}\, \tau)}{\sqrt{A}B+\sqrt{g}\cos  ( \sqrt{C}\sqrt{AB^2-g}\, \tau)} \right ), \\ \eta(\tau)&=-\sqrt{\frac{g}{A}}\left ( \frac{\sqrt{g}+\sqrt{A}B\cos  ( \sqrt{C}\sqrt{AB^2-g}\, \tau)}{\sqrt{A}B+\sqrt{g}\cos  ( \sqrt{C}\sqrt{AB^2-g}\, \tau)} \right )+B.
\end{split}
\end{equation*}

If $g=G_0=AB^2$ the same substitution yields the solution 
\begin{equation*}
\xi(\tau)=\frac{-2AB^2\tau}{ACB^2\tau^2+1}, \qquad \eta(\tau)=\frac{2B}{ACB^2\tau^2+1}.
\end{equation*}

For values of $g$ bigger than $G_0=AB^2$ one needs to consider both branches of the parametrization
to account for the two heteroclinic connections. The solution along the branch on the positive $\eta$-plane is
given by
\begin{equation*}
\begin{split}
\xi(\tau)&=-\sqrt{\frac{g}{C}}\left ( \frac{\sqrt{g-AB^2}\sinh ( \sqrt{C}\sqrt{g-AB^2}\, \tau)}{-\sqrt{A}B+\sqrt{g}\cosh  ( \sqrt{C}\sqrt{g-AB^2}\, \tau)} \right ), \\ \eta(\tau)&=\sqrt{\frac{g}{A}}\left ( \frac{\sqrt{g}-\sqrt{A}B\cosh  ( \sqrt{C}\sqrt{g-AB^2}\, \tau)}{-\sqrt{A}B+\sqrt{g}\cosh  ( \sqrt{C}\sqrt{g-AB^2}\, \tau)} \right )+B.
\end{split}
\end{equation*}
whereas the solution along the negative $\eta$-plane is given by
\begin{equation*}
\begin{split}
\xi(\tau)&=-\sqrt{\frac{g}{C}}\left ( \frac{\sqrt{g-AB^2}\sinh ( \sqrt{C}\sqrt{g-AB^2}\, \tau)}{\sqrt{A}B+\sqrt{g}\cosh  ( \sqrt{C}\sqrt{g-AB^2}\, \tau)} \right ), \\ \eta(\tau)&=-\sqrt{\frac{g}{A}}\left ( \frac{\sqrt{g}+\sqrt{A}B\cosh  ( \sqrt{C}\sqrt{g-AB^2}\, \tau)}{\sqrt{A}B+\sqrt{g}\cosh  ( \sqrt{C}\sqrt{g-AB^2}\, \tau)} \right )+B.
\end{split}
\end{equation*}

\section{Painlev\'e property of the solutions}
\label{S:Painleve}

In this section we continue analyzing the system \eqref{E:MotionNondim} under the assumption $I_{13}=0$ that, as explained before,
physically corresponds to the supposing that the
vector $\bf a$ is orthogonal to a principal axis of inertia of the body. We shall prove
\begin{theorem}
All solutions of equations  \eqref{E:Motion-simpler} 
are single valued if and only if either $I_{23}=0$, or the eigenvalue $J_1=I_{11}$ coincides with 
$J_2$ or $J_3$ (see equation \eqref{E:Eigenvalues}).
\end{theorem}
\begin{proof}
 In Section  \ref{S:a-eigenvector} we proved that if $I_{23}=0$ then
the equations become linear and homogeneous. On the other hand, in Section  \ref{SS:EqualMoments}, where 
the assumption that $J_1=J_2$ or $J_1=J_3$ was made, we
gave explicit meromorphic expressions for all the solutions. Hence, the only thing that remains
to prove is  that the system  \eqref{E:Motion-simpler} has multi-valued solutions
in the case considered in Section \ref{SS:Unequal-Moments} (where $I_{23}\neq 0$ and $J_1\neq J_2, J_3$).

  If we transform system \eqref{E:MotionNondim} into a second order
equation, then we obtain
\[
\omega_2''=\dfrac{I_{23}(\omega_2')^2}{I_{23}\omega_2+I_{33}-J_1}+\dfrac{1}{J_1I_{22}}
(J_1-I_{33}-I_{23}\omega_2)[(I_{33}-I_{22})\omega_2+I_{23}(\omega_2^2-1)],
\]
or
\begin{equation}
  \eta''=\frac{ (\eta')^2}{\eta}+\frac{1}{J_1I_{22}}\eta[-\eta^2+(I_{22}+I_{33}-2J_1)\eta+I_{23}^2 - (J_1 - I_{22}) (J_1 - I_{33})],
  \label{eq:ddW}
\end{equation}
where, just like in \eqref{E:Change-variables}, we have
\[
\eta=I_{23}\omega_2+I_{33}-J_1.
\]
We rewrite equation~\eqref{eq:ddW} in terms of the independent complex variable $\tau=z$ in the form 
\begin{equation}
\label{eq:1}
\frac{\rmd^2 \eta}{\rmd z^2}=\frac{1}{\eta} \left( \frac{\rmd \eta}{\rmd z} \right)^2+p_3 \eta^3 +p_2\eta^2 +p_1\eta. 
\end{equation}
Note that, up to a non-vanishing factor, the coefficient $p_1$ coincides with $P$ defined by  \eqref{E:DefP}. Recall that 
the condition that $P=0$ is equivalent to
saying that $J_1$ coincides with $J_2$ or $J_3$. Hence, by our analysis in Section \ref{SS:EqualMoments}
we conclude that if $p_1=0$ all solutions are single-valued. We note in passing that in this
case ~\eqref{eq:ddW} takes the form of equation XII in
Painlev\'e-Gambier classification~\cite{Ince,Gambier} that is well-known to have all
solutions single-valued. 

 We shall now prove  that if $p_1\neq 0$,  there exists a  solution to  \eqref{eq:1} with a movable logarithmic singular point.  We apply the $\alpha$-method of Painlev\'e, see Chapter XIV in \cite{Ince}. Let us introduce new variables 
\begin{equation*}
\label{eq:2}
\eta= \frac{1}{\alpha}u, \qquad z= \alpha \zeta. 
\end{equation*} 
The transformed equation reads
\begin{equation*}
\label{eq:3}
\frac{\rmd^2 u}{\rmd \zeta^2}=\frac{1}{u} \left( \frac{\rmd u}{\rmd \zeta} \right)^2+p_3 u^3 +\alpha p_2u^2 +\alpha^2 p_1 u. 
\end{equation*}
It has a solution of the form 
\begin{equation*}
\label{eq:4}
u(\zeta) = u_0(\zeta)+ \alpha u_1(\zeta) + \alpha^2 u_2(\zeta)+ \cdots, 
\end{equation*}
where the dots denote higher order terms.  
If all solutions of equation~\eqref{eq:1} are single valued, 
then $u_i(\zeta)$ must be single valued for all $i\geq 0$. The function $u_0(\zeta)$
is a solution of the equation
\begin{equation}
\label{eq:5}
\frac{\rmd^2 u_0}{\rmd \zeta^2}=\frac{1}{u_0} \left( \frac{\rmd u_0}{\rmd \zeta} \right)^2+p_3 u_0^3. 
\end{equation}
On the other hand, the functions $u_i(\zeta)$ with $i>0$ are solutions of the following linear non-homogeneous  equations
\begin{equation}
\label{eq:8}
 u_i''=\frac{2u_0'(\zeta)}{u_0(\zeta)} u_i' + \left[3p_3u_0(\zeta)^2 - \left( \frac{u_0'(\zeta)}{u_0(\zeta)} \right)^2\right] u_{i} +b_i(\zeta),
\end{equation}
where $b_1(\zeta):= p_2 u_0(\zeta)^2 $, and 
\begin{equation*}
\label{eq:9}
b_2:=  u_0 \left(p_1 +2p_2 u_1 +3p_3u_1^2 \right) + 
\frac{u_0'^2u_1^2}{u_0^3} - 2\frac{u_0'u_1 u_1'}{u_0^2} +
\frac{ u_1'^2}{u_0}. 
\end{equation*}
In general $b_i$ will  depend on $u_j$ with $j<i$.  

Let $v_1$ and $v_2$ be linearly independent solutions of homogeneous part of 
equation~\eqref{eq:8} and $V$ its fundamental matrix, i.e., 
\begin{equation*}
\label{eq:10}
V =\begin{bmatrix}
 v_1  & v_2 \\
v_1' &  v_2'
\end{bmatrix}.
\end{equation*}
Then, the  solution of~\eqref{eq:8} is given by 
\begin{equation*}
\label{eq:11}
\begin{bmatrix}
 u_i(\zeta)  \\
u_i'(\zeta) 
\end{bmatrix}=
V(\zeta)\int^{\zeta} V^{-1}(x) \begin{bmatrix}
 0  \\
b_i(x) 
\end{bmatrix}\,\rmd x .
\end{equation*}

The general solution of equation~\eqref{eq:5}  has the form
\begin{equation*}
\label{eq:6}
u_0(\zeta) = \frac{2a^2\exp[a \zeta +b]}{\exp[2(a \zeta +b)] - a^2p_3 }.
\end{equation*}
It has a simple pole at point $\zeta=c$ where $c$ is defined by the condition
\begin{equation*}
\label{eq:7}
\exp[2(a c +b)] =a^2 p_3.
\end{equation*}
It follows that if the solutions of~\eqref{eq:8} have a  branch points, then they
are located either  at $\zeta=c$ or at infinity.

It is easy check that $u_1(\zeta)$ is always single valued but
\begin{equation*}
\label{eq:12}
\operatorname{residue}\left(  V^{-1}(\zeta) \begin{bmatrix}
 0  \\
b_2(\zeta) 
\end{bmatrix}, c \right)= \begin{bmatrix}
- 2\exp(-ac-b)p_1  \\
\exp(-ac-b)p_1
\end{bmatrix}.
\end{equation*} 
Thus, if $p_1\neq 0 $ then logarithmic terms are present. This completes the proof.
\end{proof}

%
%
%
%
%
%



\begin{thebibliography}{99}
\let\\, \newcommand{\by}[1]{\textsc{\ignorespaces #1}\\}
  \newcommand{\title}[1]{\textsl{\ignorespaces #1}\\}
  \newcommand{\vol}[1]{{\bf{\ignorespaces #1}}}
  \newcommand{\info}[1]{\textrm{\ignorespaces #1}.}

\bibitem{BorMamKil} \by{Borisov A~V, Mamaev I~S and Kilin A~A} \title{Hamiltonicity and integrability of the Suslov problem}
 \info{ Regul. Chaotic Dyn. \vol{16} (2011),  104--116}
 
 \bibitem{Gambier} \by{Gambier B} \title{Sur les  \'equations diff\'erentielles du second ordre et du premier degr\'e dont l'int\'egrale g\'en\'erale est \`a points critiques fixes}
 \info{ Acta Math. \vol{33} (1910),  1--55}

\bibitem{Ince} \by{Ince, E. L.} \title{Ordinary {D}ifferential {E}quations}, \info{New York, Ordinary {D}ifferential {E}quations, 1944}

\bibitem{FMP}\by{Fedorov Y~N,  Maciejewski A~J and Przybylska M}\title{The Poisson equations in the nonholonomic Suslov problem: integrability, meromorphic and hypergeometric solutions}  \info{Nonlinearity \vol{22} (2009), 2231--2259}

\bibitem{Painleve} \by{Painlev\'e P} \title{M\'emoire sur les \'equations diff\'erentielles dont l'int\'egrale g\'en\'erale est uniforme,}
 \info{ Bulletin de la S. M. F. \vol{28} (1900), 201--261}

\bibitem{Painleve1} \by{Painlev\'e P} \title{Sur les \'equations diff\'erentielles du second ordre  et d'ordre sup\'erieur dont l'int\'egrale g\'en\'erale est uniforme}
 \info{ Acta Math. \vol{25} (1902),  1--85}

\bibitem{Suslov} \by{Suslov, G.~K} \title{Theoretical Mechanics}, \info{Moscow, Gostekhizdat, 1946 (Russian)}

\end{thebibliography}
\end{document}